\documentclass[letterpaper, 11pt, onecolumn]{article}

\pdfoutput=1

\usepackage[letterpaper,margin=1in]{geometry}
\usepackage[english]{babel}
\usepackage[utf8]{inputenc}
\usepackage{lmodern}
\usepackage[useregional]{datetime2}
\DTMusemodule{english}{en-US}

\usepackage[dvipsnames]{xcolor}
\definecolor{newcolor}{hsb}{0.6,1,0.75}

\usepackage{tikz}
\usetikzlibrary{shapes.geometric, arrows}
\usetikzlibrary{decorations.pathmorphing}
\usepackage{amsmath,amsbsy,amsfonts,amssymb,amsthm,color,dsfont,mleftright, bbm}
\usepackage{algorithm, algpseudocode}
\usepackage{enumerate}
\usepackage{hyperref}
\numberwithin{equation}{section}

\def\ddefloop#1{\ifx\ddefloop#1\else\ddef{#1}\expandafter\ddefloop\fi}

\def\ddef#1{\expandafter\def\csname b#1\endcsname{\ensuremath{\mathbb{#1}}}}
\ddefloop ABCDEFGHIJKLMNOPQRSTUVWXYZ\ddefloop

\def\ddef#1{\expandafter\def\csname c#1\endcsname{\ensuremath{\mathcal{#1}}}}
\ddefloop ABCDEFGHIJKLMNOPQRSTUVWXYZ\ddefloop

\def\ddef#1{\expandafter\def\csname t#1\endcsname{\ensuremath{\tilde{#1}}}}
\ddefloop ABCDEFGHIJKLMNOPQRSTUVWXYZabcdefghijklmnopqrstuvwxyz\ddefloop

\def\ddef#1{\expandafter\def\csname v#1\endcsname{\ensuremath{#1}}}
\ddefloop ABCDEFGHIJKLMNOPQRSTUVWXYZabcdefghijklmnopqrstuvwxyz\ddefloop

\def\ddef#1{\expandafter\def\csname v#1\endcsname{\ensuremath{{\csname #1\endcsname}}}}
\ddefloop {alpha}{beta}{gamma}{delta}{epsilon}{varepsilon}{zeta}{eta}{theta}{vartheta}{iota}{kappa}{lambda}{mu}{nu}{xi}{pi}{varpi}{rho}{varrho}{sigma}{varsigma}{tau}{upsilon}{phi}{varphi}{chi}{psi}{omega}{Gamma}{Delta}{Theta}{Lambda}{Xi}{Pi}{Sigma}{varSigma}{Upsilon}{Phi}{Psi}{Omega}{ell}\ddefloop

\def\ddef#1{\expandafter\def\csname bb#1\endcsname{\ensuremath{\bm{#1}}}}
\newtheorem{theorem}{Theorem}
\newtheorem{lemma}{Lemma}
\newtheorem{definition}{Definition}[section]
\newtheorem{remark}{Remark}

\newtheorem{problem}{Problem}

\newcommand{\gve}{\varepsilon}

\newcommand{\polyl}{\mathrm{polylog}}

\usepackage{mathtools}

\DeclarePairedDelimiter\ket{\lvert}{\rangle}
\DeclarePairedDelimiterX\braket[2]{\langle}{\rangle}{#1 \delimsize\vert #2}

\usepackage{amsmath,amssymb,amsthm,amsfonts,latexsym,bbm,xspace,thm-restate}
\usepackage{graphicx,float,mathtools,braket,mathdots}
\usepackage{booktabs,forest,soul,comment}
\usepackage{bold-extra}

\usepackage{enumitem}
\setlist{
	listparindent=\parindent,
	parsep=0pt,
}

\hypersetup{colorlinks = true,
            linkcolor = newcolor,
            urlcolor  = newcolor,
            citecolor = newcolor,
            anchorcolor = newcolor
            }

\usepackage[capitalise,nameinlink]{cleveref}

\mathchardef\mhyphen="2D 

\newcommand{\calP}{\mathcal{P}}

\newcommand{\calR}{\mathcal{R}}

\DeclareMathOperator{\Hentropy}{H}
\DeclareMathOperator{\Info}{I}

\renewcommand*{\thefootnote}{\fnsymbol{footnote}}

\begin{document}
\mbox{}
\vspace{12mm}

\begin{center}
{\huge Consumable Data via Quantum Communication}
\\[1cm] \large

\setlength\tabcolsep{2em}
\begin{tabular}{ccc}
Dar Gilboa\footnotemark &
Siddhartha Jain\footnotemark&
Jarrod R. McClean\\
\small\slshape Google Quantum AI& {\small\slshape UT Austin} & 
\small\slshape Google Quantum AI
\end{tabular}
\footnotetext[1]{darg@google.com}
\footnotetext[2]{sidjain@utexas.edu}
\date{\today}

\begin{abstract}
\noindent

Classical data can be copied and re-used for computation, with adverse consequences economically and in terms of data privacy. Motivated by this, we formulate problems in one-way communication complexity where Alice holds some data $x$ and Bob holds $m$ inputs $y_1, \ldots, y_m$. They want to compute $m$ instances of a bipartite relation $R(\mathord{\cdot}, \mathord{\cdot})$ on every pair $(x, y_1), \ldots, (x, y_m)$. We call this the asymmetric direct sum question for one-way communication. We give examples where the quantum communication complexity of such problems scales polynomially with $m$, while the classical communication complexity depends at most logarithmically on $m$. Thus, for such problems, data behaves like a consumable resource that is effectively destroyed upon use when the owner stores and transmits it as quantum states, but not when transmitted classically. We show an application to a strategic data-selling game, and discuss other potential economic implications.
\end{abstract}

\end{center}
\renewcommand*{\thefootnote}{\arabic{footnote}}
\setcounter{footnote}{0}
\section{Introduction}

As statistical models fitted to large datasets are being usefully applied to problems in various fields of science and engineering \cite{Brown2020-uq,Chen2021-hn, Merchant2023-pm}, the use of proprietary data for training or inference raises concerns of data privacy and adequate compensation for the data owner. The destructive nature of measurement in quantum mechanics has the potential to change this picture. In order to model this scenario we introduce the asymmetric direct sum question in one-way communication complexity.
Informally we say a relation $\cR$ has an asymmetric direct sum property for communication model $\cal M$ if the communication complexity of computing $\cR(x,y_1) \cdot \cR(x,y_2) \cdot \ldots \cdot \cR(x,y_m)$ is $\Omega(m)$ times the communication complexity of computing $\cR(x,y)$. One would expect examples of this in quantum communication when the state Alice sends to Bob undergoes destructive measurement, and may not be copyable. As such, we refer to problems exhibiting this property as \textit{consumable data problems}\footnote{To make our definition robust, we precisely define consumable data probltems to be those where the scaling is polyomial in $m$.}. To our knowledge, for communication complexity such a model is studied here for the first time. The work of Hazan \& Kushilevitz \cite{Hazan2017} is superficially similar but the crucial difference is that in our work Alice receives only one input whereas in theirs Alice has $m$ independent instances.

We provide a number of examples of consumable data problems when using quantum communication. For the asymmetric sum version of the Hidden Matching problem \cite{Bar-Yossef2008-uf}, we show in \Cref{sec:hm} that the one-way communication complexity is $\widetilde{\Omega}(\sqrt{m})$.
This follows from an information theoretic argument and generalizes to other settings as discussed in \cref{sec:open}.
 We conjecture this can be improved to $\Omega(m)$ for $m$ not too large, and we study a modification of the problem where we show the complexity is indeed $\Omega(m)$. Our results in \Cref{sec:mlrs,sec:no_rival_quantum_data} prove a scaling of $\widetilde{\Omega}(m)$ for the problems of sampling from a solution of a linear system and estimating the expectation values of two-outcome observables in a multi-party setting, when the parties holding the observables are restricted to using classical communication and performing single-copy measurements of Alice's messages. For these problems, we also show that classical one-way communication does not exhibit this scaling. 
More generally, the consumable data property is easily seen to \emph{not} hold for all relations when $\cal M$ corresponds to deterministic one-way communication.
Since Alice's message depends only on $x$, Bob can copy the message $m$ times and solve each of his instances. 
Additionally, the quantum asymmetric direct sum property only holds for relations which have a $\Omega(m)$ separation between quantum and deterministic communication complexity, otherwise the quantum protocol could mimic the deterministic one. 

Our results have an interesting economic interpretation, specifically in a setting where a data holder wishes to maximize the payoff of selling data that other parties wish to use for computation, or prevent unauthorized re-use. Models of markets are concerned chiefly with goods that are consumed during the process of economic production, known as \textit{rival} goods. However, for almost a century it has been recognized that data and information also play a vital role in economic processes \cite{Schumpeter1942-jr}.  The ability to cheaply replicate data has long been recognized as its chief distinguishing characteristics compared to other economic resources, and this \textit{nonrivality} has dramatic consequences \cite{Arrow1962-hf, Romer1990-sn}. It essentially implies that the (albeit idealized) equilibrium known as perfect competition, in which the price of every good on the market is set by its capacity for increasing output, cannot hold once data is involved. In some sense, one cannot ``get their money's worth'' when data is traded, unless there is some external enforcement mechanism that sets prices. Such a mechanism may lead to suboptimal resource allocations, and requires trust between the parties involved. Examples of this can be seen in recent proposals for data markets \cite{Agarwal2019-fv, Jones2020-fv}. 
Nonrivality of data may also disincentivise the creation of novel datasets, which could be of particular concern as the production rate of public high-quality data in certain modalities is far outstripped by the growth rate of training sets for large models \cite{Villalobos2022-vw}. 

In contrast to the classical picture, the fragile nature of quantum states suggests that classical data encoded in the amplitudes of a quantum state may be destroyed upon use for computation. For this to be the case, one must first show that a problem of interest can be solved with data encoded in this way. In addition, one must argue that the resulting states cannot be replicated in a similar manner to classical data. There is an inherent tension between these two goals, since while no-cloning is trivial for general quantum states \cite{Nielsen2010-jj}, this is no longer the case once states are structured. As a simple example, given a computational basis state, it can be measured in the computational basis without disturbing it and subsequently copied, and thus acts analogously to classical data. A less trivial example is that states encoding boolean functions may also be copied in some cases \cite{Aaronson2009-uz}. It is therefore a priori not obvious whether any problems satisfy these competing demands.  The nuanced nature of clonability for both the state and the task motivates the formal study of this problem.

For any relation which exhibits a quantum asymmetric direct sum property, the state(s) sent by Alice satisfy both of the requirements outlined above. 
In \Cref{sec:production_theory} we illustrate the economic consequences of this within the framework of production theory, and show that consumability implies the possibility of perfect competition, which cannot be achieved when unencrypted classical data is used. Additionally, in \Cref{sec:auction} we formulate a data market as a strategic game, and show that consumability implies potentially larger payoffs for the data seller. 

The problems we consider are based on ones that exhibit an exponential quantum communication advantage. One may wonder whether such an advantage immediately implies the asymmetric direct sum property. In \Cref{sec:no_rival_quantum_data} we show that this is not the case, by considering the problem of observable estimation in a two-party setting \cite{Kremer1995-kp, Raz1999-ic}, and using shadow tomography \cite{Aaronson2018-in}. 

\paragraph{Comparison to Chatterjee et al.~\cite{CKP25}} A recent work~\cite{CKP25} also studied a problem similar to the consumability of the Hidden Matching problem. They proved that if $k$ separated Bobs want to compute $k$ instances of Hidden Matching then the min-entropy of their joint distribution must be $\Omega(kn)$ for at least an exponentially small fraction of Alice's possible inputs. Our lower bound holds for the \emph{stronger} model where a single Bob holds all $k$ instances and our result holds in the bounded error setting, which they leave as an open problem.

\section{Related Work}

\subsection{Destructive measurement as a resource}

The idea of using uncloneability of quantum states as a feature has a long history, starting with the seminal work of Weisner \cite{Wiesner1983-lt} that introduced the notion of quantum money. However, the states used in construction of quantum money schemes typically do not encode or transmit useful information and can benefit from the computational power of pseudo-randomness in quantum state~\cite{Ji2018-cc}. While no-cloning is easy to show for states with little or no structure, this notion becomes more subtle for structured states, and in particular ones that might be useful in performing computation. Aaronson considered the question of uncloneablity of states that encode classical boolean functions, a problem he called quantum software copy-protection \cite{Aaronson2009-uz,Aaronson2021-gq}. He showed that the presence of structure enables such states to be cloned unless computational assumptions are made, and even then cannot be ruled out for states that encode functions that can be efficiently learned. The setting we consider can be seen as a distributed generalization of this problem. In the simplest case, evaluating the function of interest requires not only a quantum state in the possession of one player (or the equivalent classical description), but also an observable in the possession of another player.

The line of work which most closely resembles our work conceptually is that of One-Time Programs, introduced by Goldwasser et al.~\cite{Goldwasser2008-gy}. The quantum analog of these was defined by Broadbent et al.~in 2013 and has recently seen a surge of interest \cite{BGS13,BenDavid2023quantumtokens,CLLZ21,cryptoeprint:2024/1798,cryptoeprint:2024/1934}. The main difference in our definition is that it is information theoretic and thus does not rely on any computational assumptions. We note that a natural definition of One-Time Programs is provably impossible in many cases, and recent works weaken the model to make it more generically feasible. We hence make an orthogonal inquiry into what is achievable in a very strong sense, even if only in special cases.



\subsection{Secure Multi-Party Computation}
\label{sec:mpc}

The ability to prevent the re-use of data for computation can in principle be achieved classically using the tools of secure multi-party computation (MPC). The principal objective of line of work is the evaluation of a function $f(x,y)$ where Alice holds $x$ and Bob holds $y$, in a manner that ensures the security of each player's input and reveals only $f(x,y)$ to both players. There has been extensive work on this problem in various forms since its formulation by Yao \cite{Yao1982-yt, Yao1986-rw} (see \cite{Evans2018-py} for a review). Elegant solutions to this problem are known that involve obfuscating (or garbling) the circuit describing $f$ one gate at a time so as to obscure the inputs of each player~\cite{Goldreich1987-na} (which can be achieved using standard cryptographic primitives such as public-key cryptography) or alternatively based on fully homomorphic encryption~\cite{Asharov2012-dq}. Using MPC, the players can run a protocol that enables the evaluation of $f(x,y)$ for a single pair of inputs. However, if Bob wanted to evaluate $f(x,y')$ for some $y' \neq y$, the validity of any MPC scheme implies that this would generally be impossible without rerunning the protocol. Since this requires the cooperation of both parties, it could allow one party to control the number of times another party can compute functions of their joint inputs. 

MPC is incomparable to the consumability of data studied in this work, which relies on the properties of quantum mechanics. MPC has the benefits of being generic and not requiring the constant overheads associated with fault-tolerant quantum computation. However, MPC has a number of drawbacks compared to consumable data, namely (i) it requires multiple rounds of two-way communication~\cite{Gennaro2001-xw} whereas our notion of consumable data requires only a single round of one-way communication, (ii) it requires cryptographic assumptions while we give unconditional results and (iii) it requires coordination between parties e.g. in choosing a cryptosystem to use, while our construction requires no such coordination. The best known classical techniques also have overheads associated with them due to the need to encrypt data, which however may not be fundamental. 

\section{Preliminaries \& Notation}


We denote by $D^\rightarrow$ deterministic classical one-way communication complexity. $R^\rightarrow_\gve$ denotes randomized one-way classical communication complexity with error probability at most $\gve$, in which players are allowed to share an unlimited number of public random bits that are independent of their inputs. We similarly define by $Q^\rightarrow_\gve$ one-way quantum communication complexity with error probability at most $\gve$. In all cases the one-way restriction implies that only Alice is allowed to send messages to Bob (if there are multiple Bobs, they can communicate among themselves and we do not consider this as part of the complexity of the problem). When the error is a nonzero constant (say $1/3$) we omit the subscript. For formal definitions we refer the reader to textbooks by Nisan \& Kushilevitz \cite{Kushilevitz_Nisan_1996}, and Lee \& Shraibman \cite{Lee09}.

We also consider sampling problems, where the goal is for Bob to produce a sample from a target distribution (or some distribution close to it) given some inputs to Alice and Bob.  For this type of problem, we define analogously $SR, SQ$ for the classical (randomized) and quantum communication complexity respectively (with the superscript $\rightarrow$ denoting one-way communication as before).


\begin{definition}[TV Distance]
    The total variation distance between two distributions $p,q$ supported on $\cal X$ is given by
    $$d_{TV}(p,q) = \sup_{S \subseteq {\cal X}} |p(S) - q(S)|$$
\end{definition}

When we consider sampling problems, we allow constant error in TV distance between the target distribution and the one sampled by the algorithm. Finally, we denote by $A^+$ the pseudoinverse of $A$.

\section{Consumable data} \label{sec:consumable_data}

We now define the notion of consumable or rival data. Denote by $\cX, \cY, \cO$ the space of Alice's inputs, Bob's inputs (or those of a single Bob in case there is more than one), and a space of outputs. Below, we use $P = (\cR, \mathcal{P}_P, q)$ 
to denote a family of relational problems $\cR \subseteq \cX \times \cY \times \cO$ and a set of protocols $\mathcal{P}_P$. We informally use problem to refer to tuples of this kind. Note that one can construct a similar definition for sampling problems, where for each input the goal is for Bob to output a sample from a specific distribution. 

We use $\cR^m \subseteq \cX \times \cY^m \times \cO^m$ to denote the $m$-Bob relational problem where Alice's input is kept constant but all the $m$ Bobs have distinct inputs. 
The \emph{goal} is to solve the relation on all of Bob instances with 2/3 probability\footnote{If $\mathcal{R}$ was a decision problem that could be solved with failure probability $\delta$, one could solve $\mathcal{R}^m$ with failure probability $\delta$ as well by simple repetition, incurring a multiplicative overhead logarithmic in $m$. However, this is no longer the case when considering relations, so this notion of complexity is not finite in general.}. Similarly, we use $\mathcal{P}^m_P$ to denote the set of protocols where Alice sends one message and the Bobs are allowed to communicate classically if $q=0$ and quantumly if $q=1$. So $P^m = (\cR^{m}, \mathcal{P}^m_P, q)$.

For a problem $P$, denote $c(P)$ to be the one-way communication complexity of the minimum cost protocol in the set $\mathcal{P}_P$ which solves $\cR$. Since we will be modeling scenarios where Alice is selling her data to the Bobs who will be using it for computation, the cost here will be in terms of communication between Alice and the Bobs only. 

\begin{definition}[Consumable data problem] \label{def:consumable}
    A problem $P$ is said to be a consumable data problem if as $m \rightarrow \infty$
    \[
    \frac{c(P^m)}{c(P)} = m^{\Omega(1)}
    \]
\end{definition}

\begin{definition}[Non-consumable data problem]
    A problem $P$ is said to be a non-consumable data problem if as $m \rightarrow \infty$
    \[
    \frac{c(P^m)}{c(P)} = m^{o(1)}
    \]
\end{definition}

We refer to the quantity appearing in the lower bound in \Cref{def:consumable} as the \textit{consumability rate} of $P$. In other words, we can say that a problem $P$ is a consumable data problem if its consumability rate is polynomial, and it is non-consumable otherwise. Note that $m$ is bounded by a function of $n$ so we have a slight abuse of notation, but this will not be a concern for the problems we consider. There is a subtlety in this definition, in the sense that the benefit of consumability arises when Alice chooses to use a particular communication protocol (typically a quantum one over a classical one) but the definition itself does not specify why she would have such a preference. A natural way to introduce a preference is by formulating a strategic game that involves communication. We provide an example in \cref{sec:auction}. 


Note that any problem involving only deterministic classical communication must be non-consumable -- every Bob can just copy Alice's message into his own working space. We also show in \Cref{sec:no_rival_quantum_data} that if $P$ corresponds to a \emph{decision} problem, then even with quantum communication it must be a non-consumable data problem. This is because the Bobs can apply the Shadow Tomography protocol \cite{Aaronson2018-in} (unless the Bobs are only allowed classical communication between themselves and limited quantum memory). Nevertheless, consumability can be proved for certain search problems (with many solutions) solved using randomised or quantum communication. There are a few cases where consumability or nonconsumability can be characterized, which we discuss below:

\begin{lemma}
\label{lem:deterministic}
    For any relational problem $\cR$ and resource $q$, if the protocol is deterministic one way classical communication, $P = (\cR, D^{\rightarrow},q)$, then $\frac{c(P^m)}{c(P)} = 1$ and the data is non-consumable.
\end{lemma}

\begin{proof}
\label{pf:deterministic}
    For the $m$-Bob problem, Alice sends the \emph{same} message as the protocol for the original problem. Since her message depended only on her input, and must enable Bob to solve the problem for any possible input of his, the message can be re-used $m$ times and the correctness guarantee holds for every instance on Bob's end. 
\end{proof}

 Similarly, 
\begin{lemma}
\label{lem:random}
    For any relational problem $\cR \subseteq \cX \times \cY \times \cO$, with $|\cO| = K$ and resource $q$, if the protocol is randomized one-way classical communication, $P = (\cR, R^{\rightarrow},q)$, then $\frac{c(P^m)}{c(P)} = O(K\log m)$.
\end{lemma}
\begin{proof}
\label{pf:random}
    This can be achieved by learning the distribution of Bob's output under the randomness of Alice's message. Using folklore results (see \cite{canonne2020shortnotelearningdiscrete}) this can be done using $K \log m$ times the amount of communication, by learning the $K$-outcome distribution up to error $1/m$.
\end{proof}

\begin{lemma}
\label{lem:quantum}
    For any relation with an output space of size K, $\cR \subseteq \cX \times \cY \times \cO$, with $|\cO| = K$, if the protocol is one way quantum communication $P = (\cR, Q^{\rightarrow}, q=1)$ then $\frac{c(P^m)}{c(P)} = \widetilde{O}(K \log^2 m) c(P)$.
\end{lemma}

\begin{proof}
\label{pf:quantum}
    Akin to \Cref{lem:random}, we want to give a protocol for for $P^m$ using a protocol for $P$ as a subroutine. We do this by relying on the work of Gong \& Aaronson \cite{pmlr-gong23a} who proved that the distribution of $K$-outcome POVMs on $\log N$ qubits can be learned to constant additive error in $\widetilde{O}(K \log^2 m \log N)$ copies.
\end{proof}

All of these lemmas can be generalized to the setting where $\mathcal{P}_P$ is a strict subset of one of these sets, which applies to the sampling models. Note that the setting where Alice only sends \emph{samples} of a quantum state is the setting of shadow tomography~\cite{Aaronson2018-in}, so our lower bounds generalize lower bounds for shadow tomography.


We also define a related notion of \emph{strongly} consumable data. To do this, we define the problem $SP^m = (\calR^m, \calP_{SP}^m, q)$ where $\calP_{SP}^m$ is the set of protocols which solve a $2/3$ fraction of Bob's $m$ instances successfully with high probability. Note that this makes Bob's task easier, hence achieving strong consumability is harder from Alice's perspective.

\begin{definition}[Strongly consumable data problem] \label{def:s-consumable}
    A problem $P$ is said to be a consumable data problem if
    \[
    \frac{c(SP^m)}{c(P)} = m^{\Omega(1)}
    \]
\end{definition}

Note that in the strongly consumable case, all classical protocols are non-consumable.

\begin{lemma}
\label{lem:s-random}
    For any relational problem $\cR \subseteq \cX \times \cY \times \cO$, with $|\cO| = K$ and resource $q$, if the protocol is randomized one-way classical communication, $P = (\cR, R^{\rightarrow},q)$, then $\frac{c(SP^m)}{c(P)} = O(\log m)$ and the data is not strongly consumable.
\end{lemma}

Despite this classical impossibility, our results in the next section continue to hold for strongly consumable data in the quantum setting.

\section{Examples of consumable quantum data} \label{sec:examples}

\subsection{Hidden Matching}
\label{sec:hm}

Hidden matching is a famous example of a relation that exhibits an exponential separation between quantum and classical one-way communication complexity \cite{Bar-Yossef2008-uf}. We prove that for the asymmetric direct sum version of the Hidden Matching problem, quantum data behaves as a consumable resource (while classical data does not). 

The original problem is defined as follows:

\begin{problem}[Hidden Matching \cite{Bar-Yossef2008-uf}]
    Alice is given a string $x \in \{0, 1\}^N$. Bob is given a perfect matching $M$ over $[N]$. Their goal is for Bob to output $(i,j,x_i \oplus x_j)$ where $(i,j) \in M$. Only Alice is allowed to send messages to Bob. 
\end{problem}

One can naturally generalize this problem to the setting of multiple matchings as follows:

\begin{problem}[Multiple Hidden Matchings ($\mathrm{MHM}_{N,m}$)]
        Alice is given a string $x \in \{0, 1\}^N$. Each of the $m$ Bobs is given $m$ perfect matchings $\{M_k\}$ over $[N]$. Their goal is to output $(i,j,x_i \oplus x_j)$ where $(i,j) \in M_k$ for all $k$. Only Alice is allowed to send messages to Bob.
\end{problem}

A tight lower bound shows that classical communication indeed acts like a nonrival good for this problem. While it is known that $R^\rightarrow(\mathrm{HM}_{N}) = \Omega(\sqrt{N})$ \cite{Bar-Yossef2008-uf}, we believe this is the first characterization of the deterministic complexity of the Hidden Matching problem. The results are consistent with \Cref{lem:deterministic}, and provide an explicit demonstration of the phenomenon where a message large enough to solve a single instance of a relation can be re-used to solve others. 

    \begin{lemma} \label{lem:MHM_Dlb}
        $D^\rightarrow(\mathrm{MHM}_{N,m}) = D^\rightarrow(\mathrm{HM}_{N}) = N/2+1$.
    \end{lemma}
    Proof: \Cref{pf:lem:MHM_Dlb}. 

Now we note that even with randomized communication, this relation does not possess the consumable data property.

    \begin{lemma} \label{lem:MHM_Rub}
        $R^{\rightarrow}(\mathrm{MHM}_{N}) = O(\sqrt{N}\log m)$
    \end{lemma}
    \begin{proof}
        We adapt the upper bound for $\mathrm{HM}_{N}$. Alice sends the values of randomly chosen $O(\sqrt{N} \log m)$ nodes, which by a birthday paradox style calculation and union bound has a constant probability of containing one edge from each matching.
    \end{proof}

There is a quantum algorithm that solves this problem with probability $1$ using $m\log N$ qubits of communication, which is a trivial repetition of the algorithm of \cite{Bar-Yossef2008-uf}: 
\begin{lemma} \label{lem:MHM_Qub}
    $Q^\rightarrow(MHM_{N,m}) = O(m \log N)$.
\end{lemma}
\begin{proof}
    Alice sends Bob a copy of the state $\ket{\psi}=N^{-1/2}\sum_{i=1}^{N}(-1)^{x_{i}}\left|i\right\rangle $ over $\log N$ qubits. Denoting the $k$-th pair in Bob's a matching that Bob holds by $(i_k,j_k)$, Bob measures the state using the $N$-outcome POVM defined by $E_{k,b}=\frac{1}{2}\left(\left|i_{k}\right\rangle +(-1)^{b}\left|j_{k}\right\rangle \right)\left(\left\langle i_{k}\right|+(-1)^{b}\left\langle j_{k}\right|\right)$
for $k \in [N/2], b \in \{0, 1\}$. This process is repeated for every matching. 
\end{proof}
It is clear that the states in the algorithm above cannot be re-used after a measurement to solve the problem for multiple matchings. Since each POVM has $N$ possible outcomes, approaches based on gentle measurement that are discussed in \Cref{sec:no_rival_quantum_data} should not be applicable to this problem without requiring $\mathrm{poly}(N)$ copies of the state.

We also have the following lower bound on the quantum communication required to solve the problem.

\begin{lemma} \label{lem:MHM_QLB}
 $Q^\rightarrow(MHM_{N,m}) = \Omega(\sqrt{m})$ for $m \leq N/2.$
\end{lemma}
\begin{proof}
\label{pf:MHM_QLB}
    Let us consider the distributional complexity of $\mathrm{MHM}_{N,m}$ where Alice's input is a uniform random string $\mathbf{X} \sim U(\{0,1\}^N)$. The Bobs have a deterministic input Y, where $M_1$ is just the matching $\{(i, i+1) | i \text{ odd}, i < N\}$. The matching $M_k$ is just the $k^{th}$ cyclic permutation on nodes on the left. The Bobs output random variables $o_k = (i_{k}, j_{k}, x_{i_k} \oplus x_{j_k})$ as their respective solutions. For notational convenience, we define $\mathbf{O} = o_1o_2\ldots o_m$. Note that since $m \leq N /2$, each matching consists of $N/2$ edges that do not appear in any other matching. It follows that for any choice of $\mathbf{O}$, no edge (as defined by the first two entries of each $o_k$) will be repeated.

    Let $\rho_{\mathbf{X}}$ be density matrix corresponding to the message of length $l$ sent by Alice, of dimension $2^l$. By Holevo's theorem, $\Info(\mathbf{X}:\mathbf{O}) \leq l$. We will show that if the Bobs solve $\mathrm{MHM}_{N,m}$ then $\Info(\mathbf{X}:\mathbf{O}) \geq \Omega(\sqrt{m})$. This gives us the required lower bound. 
    
    $\Info(\mathbf{X}:\mathbf{O}) = \Hentropy(\mathbf{X}) - \Hentropy(\mathbf{X} | \mathbf{O}) = n - \Hentropy(\mathbf{X} | \mathbf{O})$. To make the random variable $\mathbf{O}$ amenable to analysis, we remove dependencies in the output by considering a spanning forest of the graph induced by $V = \bigcup_k \{i_k, j_k\} = \cup_k \{i_k\} \bigcup \cup_k \{j_k\}$. We have that $|V| \geq \Omega(\sqrt{m})$ since we have a graph with $m$ distinct edges by construction. Take a spanning tree (which has size $\Omega(\sqrt{m})$) and call the collection of these random variables $\mathbf{O_T}$. We can write $\Hentropy(\mathbf{X} | \mathbf{O}) = \Hentropy(\mathbf{X} | \mathbf{O_T} \mathbf{O'}) \leq \Hentropy(\mathbf{X} | \mathbf{O_T})$.

    Now, note that parities encoded in $\mathbf{O_T}$ introduce $k = \Omega(\sqrt{m})$ binary constraints on the random variable $\mathbf{X}$ over the boolean hypercube, each of which reduces the support of the conditional distribution by a factor of $2$. Thus, we have that $\Hentropy(\mathbf{X|O_T}) \leq \log(2^{n - k}) = n-k$ which means that $\Info(\mathbf{X}:\mathbf{O}) \geq k = \Omega(\sqrt{m})$. This gives us the desired lower bound by Holevo's theorem.
\end{proof}

Combining \Cref{lem:MHM_Qub} and \Cref{lem:MHM_QLB} gives for $m \leq N/2$
\begin{equation}
    \frac{Q^{\rightarrow}(\mathrm{MHM}_{N,m})}{Q^{\rightarrow}(\mathrm{MHM}_{N,1})}=\tilde{\Omega}(\sqrt{m}).
\end{equation}
This implies that Multiple Hidden Matchings is a consumable quantum data problem, with consumability rate $\approx \sqrt{m}$.
Note that this is not the case classically. The deterministic lower bound in \Cref{lem:MHM_Dlb} also illustrates explicitly that the message sent by Alice to solve a single matching, if composed of raw bits of her input, can be immediately re-used to solve the problem for all possible matchings (since $N+1$ bits will contain the end-points of an edge of any possible perfect matching).

\begin{remark}
    We note that in the proof of \cref{lem:MHM_QLB}, even if $\Omega(m)$ Bobs outputted an answer, we would have that $|V| \geq \Omega(\sqrt{m})$. Therefore, the lower bound holds in the average case and Multiple Hidden Matchings is a strongly consumable data problem with quantum communication.
\end{remark}

Finally, we define a problem with a slightly differrent definition which does \emph{not} formally fit into the framework of Consumable Data (\cref{def:consumable}). However, it is conceptually similar and we can achieve a tight lower bound for it.

\begin{problem}[Hidden Matching with Many Edges ($\mathrm{HMM}_{N,m}$)]
    Alice is given a string $x\in \{0,1\}^n$. Bob is given a perfect matching $M$ over $[N]$. Their goal is for Bob to ouput $m$ tuples $(i_k, j_k, x_{i_k} \oplus x_{j_k})$ where $(i_k, j_k) \in M$ for all $i \in [m]$.
\end{problem}

We now show that the communication complexity of this problems scales linearly in $m$.

\begin{lemma}
    $Q^\rightarrow(\mathrm{HMM}_{N,m}) = \Omega(m)$ for $m\leq N$
\end{lemma}
\begin{proof}
    We define the distribution to be $\mathbf{X} \sim U(\{0,1\}^n)$ with a single perfect  matching $M$. Then define the random variable $\mathbf{O}$ in the same way as in the proof of \cref{lem:MHM_QLB}, but now we are guaranteed that $|V| \geq m$, because we are requiring that Bob outputs distinct edges from a matching. This implies that $I(\mathbf{O} : \mathbf{X}) \geq \Omega(m)$, giving us a lower bound of $\Omega(m)$ by Holevo's theorem.
\end{proof}

\subsection{Linear regression sampling}
\label{sec:mlrs}

Another key problem type for which it is possible to transform data into a rival good is sampling problems with a quantum communication advantage.  In this type of problem, Alice sends Bob a message, which Bob uses to sample from a target distribution with high accuracy.  The essence of the construction is that the quantum communication advantage allows Alice to reveal only a tiny fraction of the original data while allowing Bob to solve the problem, and the method by which he solves it destroys the data that was sent, not allowing it to be reused to generate more samples.  We consider here a sampling  variant of linear regression introduced by Montanaro et al. \cite{Montanaro2024-qj}:
\begin{problem}[Linear Regression Sampling  \cite{Montanaro2024-qj} ($\mathrm{LRS}_{N}$)]
	Alice is given a vector $x \in \bS^{N-1}$. Bob is given a matrix $B$. The goal is for Bob to produce a sample from the distribution $\cP$ over $[N]$ defined by 
	\begin{equation} 
		p_{i}=\left|\left[B^{+}x\right]_{i}\right|^{2}/\left\Vert B^{+}x\right\Vert _{2}^{2}.
	\end{equation}
\end{problem}

One can naturally generalize this problem to the setting of multiple samples as follows:
\begin{problem}[Multiple Linear Regression Sampling ($\mathrm{MLRS}_{N,m}$)]
    Alice is given a vector $x \in \bS^{N-1}$. Bob is given $m$ matrices $B_k$. The goal is for Bob to produce one sample from each distribution $\cP_k$ over $[N]$ defined by 
    \begin{equation} \label{eq:MLRS_dist}
        p^{(k)}_{i}=\left|\left[B_{k}^{+}x\right]_{i}\right|^{2}/\left\Vert B_{k}^{+}x\right\Vert _{2}^{2}.
    \end{equation}
    
\end{problem}
    Note that solving the above problem with some inaccuracy $\eta$ corresponds to sampling from some distribution with total variation error at most $\eta$ with respect to $\cP_k$. 
    In order to consider the communication complexity of these problems, we must first discretize the inputs so that they have finite size. We thus assume all real number are specified to $\log N$ bits of precision. 
We then have the following lemma for classical protocols.
\begin{lemma}[\cite{Montanaro2024-qj}] \label{lem:MLRS_cb}
    For constant total variation distance error $\eta$ in the sampled distribution,
    \begin{enumerate}[label=\roman*), align=left]
        \item $SR^\rightarrow_\eta(\mathrm{MLRS}_{N,1}) = \Omega(N \log N)$.
        \item For any $m$, $SR^\rightarrow_\eta(\mathrm{MLRS}_{N,m}) = O(N \log N)$.
    \end{enumerate}
\end{lemma}

Meanwhile for quantum protocols, we have a linear dependence on $m$.

\begin{lemma} \label{lem:MLRS_qb}
We have the following bounds on $SQ_{\eta}^{\rightarrow}$,
    \begin{enumerate}[label=\roman*), align=left]
        \item For TV error $\eta \leq 1/4$, $SQ_{\eta}^{\rightarrow}(\mathrm{MLRS}_{N,m})=\Omega(m\log(N/m))$.
        \item For constant TV error $\eta$, $SQ_{\eta}^{\rightarrow}(\mathrm{MLRS}_{N,m})=O(m\log(N)\underset{k}{\max}(\left\Vert B_{k}^{+}\right\Vert ^{2}/\left\Vert B_{k}^{+}x\right\Vert _{2}^{2}))$.
    \end{enumerate}
\end{lemma}

We prove this in \Cref{pf:MLRS_qb}.

While these upper and lower bounds match in terms of their $N$ dependence if $\left\Vert B_{k}^{+}x\right\Vert _{2}$ is relatively large (and in particular does not decay with $N$), they do not match in terms of their $m$ dependence. One example is when the features of $x$ that different samples are sensitive to are in some sense uniformly distributed, as in the construction used to obtain the lower bound in \Cref{lem:MLRS_qb}. In this case, we have $\underset{k}{\max}\left\Vert B_{k}^{+}\right\Vert^2/{\left\Vert B_{k}^{+}x\right\Vert^2 _{2}} = O(m)$. It follows that, restricting to such inputs, we have 
\begin{equation}
    \frac{SQ_{1/4}^{\rightarrow}(\mathrm{MLRS}_{N,m})}{SQ_{1/4}^{\rightarrow}(\mathrm{MLRS}_{N,1})}=\tilde{\Omega}(m).
\end{equation}
Based on the definitions of \Cref{sec:consumable_data}, we obtain that $\mathrm{MLRS}_{N,m}$ is a consumable data problem for quantum data, with consumability rate $m$.

\subsection{Decision problems} \label{sec:no_rival_quantum_data}

In the previous examples, we considered problems that exhibit an exponential quantum communication advantage. It is natural to ask if such an advantage implies consumability in some generic sense.  We will see that this is not the case when Bob's task is a decision problem. 
  
The examples we discuss here are based on the following problem:
\begin{problem}[Vector In Subspace ($\mathrm{VS}_{N,\theta}$) \cite{Kremer1995-kp}] 
    Alice is given a vector $x \in \bS^{N-1}$. Bob is given two orthogonal subspaces of dimension $N/2$ specified by projection operators $M^{(1)},M^{(2)}$. Under the promise that either $\left\Vert M^{(1)}x\right\Vert _{2}\geq\sqrt{1-\theta^{2}}$ or $\left\Vert M^{(2)}x\right\Vert _{2}\geq\sqrt{1-\theta^{2}}$ for $\theta < 1/\sqrt{2}$, determine which is the case. 
\end{problem}

It is known that Vector in Subspace exhibits an exponential advantage in quantum communication with respect to randomized classical communication complexity \cite{Regev2011-wq}. Consider the following generalization:

\begin{problem}[Vector In Multiple Subspaces ($\mathrm{VMS}_{N,\theta,m}$)] 
    Alice is given a vector $x \in \bS^{N-1}$. Bob is given $m$ pairs of orthogonal subspaces $M^{(1)}_j,M^{(2)}_j$. Given a similar promise to the vector in subspace problem for each pair of subspaces, the goal is to determine which subspaces $x$ has large overlap with. 
\end{problem}

The exponential advantage in quantum communication might suggest that for this problem as well, classical data will behave like a nonrival good while the quantum analog might behave like a rival good.  
This is because even for $m=1$, Alice must send a significant portion of her input to Bob, and thus she may not be able to derive value that is polynomial in $m$ for larger $m$. However, the problem can still be solved with relatively little quantum communication, since data states can be re-used in a manner that allows Bob to solve the problem for $m>1$ with Alice communicating a number of qubits that is only logarithmic in $m$. This can be achieved via shadow tomography:

\begin{theorem}[Shadow Tomography \cite{Aaronson2018-in} solved with Threshold Search~\cite{Badescu2021-kb}] \label{thm:shadow_tomography}
    For an unknown state $\ket{\psi}$ of $\log N$ qubits, given $m$ known two-outcome measurements $E_i$, there is an explicit algorithm that takes $\ket{\psi}^{\otimes k}$ as input, where $k = \tilde{O}(\log^2 m \log N \log(1/\delta)/ \gve^4 )$, and produces estimates of $\left\langle \psi\right|E_{i}\left|\psi\right\rangle $ for all $i$ up to additive error $\gve$ with probability greater than $1 - \delta$. $\tilde{O}$ hides subdominant polylog factors.
\end{theorem} 

$\mathrm{VMS}_{N,\theta,m}$ is a problem of estimating $m$ expectation values up to some constant error (due to the constraint on $\theta$) on a target state. If polynomial error is required, it is known that $\Omega(N)$ qubits of communication may be needed, and hence quantum communication is essentially equivalent to classical communication (from e.g. lower bounds on estimating inner products \cite{Cleve1999-dy}). Allowing constant error, \Cref{thm:shadow_tomography} says that $\mathrm{polylog}(m)$ qubits of communication suffice to solve the problem. This directly implies that, at least if Alice sends multiple copies of her state, a lower bound analogous to \Cref{lem:MLRS_qb} is impossible, as Bob does not require a number of qubits polynomial in $m$. This shows that an exponential communication advantage is not a sufficient condition for quantum data to behave like a rival good. 

Given that multiple entangled copies of a quantum state are known to be a more powerful resource than single copies \cite{Huang2022-ot}, it would also be interesting to consider a setting where Alice sends only single copies of her data states. One way to do this is by introducing assumptions about the computation Bob is allowed to perform with his message. It may be possible to remove this assumption by utilizing certified deletion \cite{Broadbent2020-xw}. While requiring additional encryption, this could enable Alice to only send a copy of her state after receiving a certificate that Bob has deleted the previous copy, ensuring that multi-copy measurements cannot be performed. 

A key difference between the Vector-in-Subspace problem and the other problems we consider is that the former is a decision problem (a two-outcome measurement), while the latter are sampling problems or relations. This difference was already captured by \Cref{lem:quantum}, where we showed that if the number of outcomes are small then the problem does not exhibit consumability for a large range of parameters. In the next section, we get around this limitation by considering the multiparty setting.


\subsection{Multiple Bobs: A communication arms race}

The above picture changes when more than two parties are involved. Consider a setting where Alice has a vector which she can encode in a quantum state $\ket{x}$ and each of $m$ Bobs has an observable $O_i$, Alice is only willing to send the Bobs copies of $\ket{x}$ (when using quantum communication), and the Bobs cannot (i) store multiple copies of $\ket{x}$ or (ii) communicate quantum states between them, this is equivalent to the setting of learning without quantum memory that is studied in \cite{Chen2022-kx}. More precisely, this is a setting where each Bob can perform a POVM on a single copy of $\ket{x}$ only, and exchange classical messages which correspond to the classical memory used in this setting. In contrast, the setting of learning with quantum memory (as per \cite{Chen2022-kx}) is one where the Bobs are allowed quantum communication (but still can measure only a single copy of $\ket{x}$ each), with the content of the quantum communication channel corresponding to the quantum memory. In both cases, Alice's messages correspond to samples of a quantum state (unknown to Bob) as is standard in learning problems. While the results of \cite{Chen2022-kx} apply to samples of a mixed state described by a density matrix $\rho$, they also apply to a purification of $\rho$ in a larger space. This will not affect the scaling with $m$ which is the main object of interest for our purposes. 

Define by $\cO$ an ensemble of two-outcome POVMs given by $O_{i}=U_{i}Z_{n}U_{i}^{\dagger}$ for $0\leq i<m/2$ and $O_i=-U_{i-m/2}Z_{n}U_{i-m/2}^{\dagger}$ for $m/2\leq i<m$, where the $U_i$ are drawn i.i.d. from the Haar measure and $Z_n$ acts only on the last qubit. 

When only classical communication is used between Alice and the Bobs, an optimal lower bound of $\Omega(\sqrt{N})$ for estimating the expectation value of a single two-outcome observable with constant probability is applicable \cite{Gosset2019-di}. Lemma 1 of that paper also provides a matching upper bound in the $m$-observable case (up to logarithmic factors). Namely, estimating $m$ expectation values of unit norm observables to constant error can be done with probability $2/3$ by sending $\tilde{O}(\log(m)\sqrt{N})$ bits from Alice to Bob (where $\tilde{O}$ hides $\polyl(N)$ factors). Alice requires no knowledge of the observables themselves. This protocol is based on sending $O(\log(m))$ random stabilizer sketches of Alice's input state $\ket{x}$. Each sketch involves Alice drawing a Clifford unitary $C$ from a uniform distribution over the Clifford group $\cC_n$ ($n=\log N)$, and computing $\bigl\langle0^{\otimes(n-k)}z\bigr|C\bigl|x\bigr\rangle $ for all $z \in \ket{0,1}^k$ for $2^k = \tilde{O}(\sqrt{N})$. Alice generates $O(\log(m))$ i.i.d. sketches in this way and sends both the measurement results and a description of the Clifford unitaries to the Bobs. Each Clifford unitary is defined by specifying $O(n^2)$ one or two-qubit gates from a small set, and thus has an efficient classical description. 

If Alice instead sends copies of her input encoded in the amplitudes of a quantum state $\ket{x}$ to the Bobs, but we allow classical communication only between the Bobs, and restrict the Bobs to performing single-copy measurements, the number of samples of $\ket{x}$ required is linear in $m$ \cite{Chen2022-kx}:

\begin{theorem}[Corollary 5.7, \cite{Chen2022-kx}]
    With constant probability over $O_i$ drawn i.i.d. from $\cO$, estimating the expectation values of all $O_i$ w.r.t. $\ket{x}$ without quantum communication between Bobs with success probability at least $2/3$ requires $\Omega\left(\min\left\{ m/\log(m),N\right\} /\varepsilon^{2}\right)$ copies of $\ket{x}$. 
\end{theorem}


\begin{table}[t]
    \centering
    \renewcommand{\arraystretch}{1.5}
    \begin{tabular}{l c}
         & \textbf{(Qu)bits sent from Alice to the Bobs}\\
        \hline
        Classical A $\rightarrow$ Bs, Classical Bs $\leftrightarrow$ Bs & $\tilde{\Theta}(N^{1/2})$ \cite{Gosset2019-di}\\
        \hline

        Quantum A $\rightarrow$ Bs, Classical Bs $\leftrightarrow$ Bs & $\tilde{\Theta}(m \log N)$ \cite{Huang2020-gf}\\
        \hline
        Quantum A $\rightarrow$ Bs, Quantum Bs $\leftrightarrow$ Bs & $O((\log(m) \log(N))^2)$ \cite{Badescu2021-kb} \\
    \end{tabular}
    \caption{A communication arms race in estimating expectation values of two-outcome observables to within constant error: Data behaves as a consumable resource if Alice is only willing to send quantum states encoding her data, while the Bobs can only communicate classically. This ceases to be the case if only classical communication is used, or if the Bobs can communicate quantum states. $\tilde{\Theta}(\cdot)$ hides factors of $\log m$.}
    \label{tab:two_outcome}
\end{table}

Note that this is worst-case over $\ket{x}$ (if $\ket{x}$ was uniformly random Bobs could just guess $0$). Note also that the $O_i$ are chosen so that classical shadows do not help (for the operators in question the Hilbert-Schmidt norm is $||O^2_i||=N$, which is roughly equivalent to the shadow norm that sets the sample complexity of classical shadows \cite{Huang2020-gf}). A matching upper bound (up to $\log(m)$ factors, as long as $m < N$) is obtained by the straightforward approach in which Alice sends each Bob $O(1/\gve_2)$ copies of her state.  

When the Bobs are allowed to use quantum communication, we are essentially back to the two-party version of the problem, since they can jointly use shadow tomography \cite{Aaronson2018-in, Badescu2021-kb} to estimate all the expectation values using a logarithmic number of copies of $\ket{x}$. These results are summarized in \Cref{tab:two_outcome}. 

\section{Economic implications of consumable data} \label{sec:economics}
We now illustrate the economic implications of the complexity-theoretic asymmetric direct sum property. We give some background on the effect of the ability to copy classical data on economic models, and, more concretely, design an auction where the buyer is forced to purchase $m$ copies of the data if they want to use it $m$ times, increasing the payoff of the data seller compared to an analogous situation when selling classical data.

\subsection{Data as an economic resource in production theory} \label{sec:production_theory}

Production theory \cite{Kurz1995-yj} is one of the principal frameworks for the quantitative study of economic systems. A fundamental object of interest within this framework is the \textit{production function} $F:\mathbb{R}_{+}^{M}\rightarrow\mathbb{R}_{+}$ that quantifies in some form the output of an economic agent, for example the goods produced by a firm. The inputs to $F$ denote the resources required to produce said goods, such as labor, capital and raw materials. For conventional goods of this form, which cannot be replicated at zero cost (and are referred to as \textit{rival} goods), it is known that the production function is typically a degree $1$ homogeneous function of its inputs (at least locally when restricted to some set $S$):
\begin{equation} \label{eq:prod_function}
    F(\lambda x)=\lambda F(x)
\end{equation}
for any $\lambda \geq 0$\footnote{Strictly speaking, this relationship holds only if each good can serve as a substitute for another, which is a standard assumption.}. This captures the notion that e.g. doubling the number of raw materials will double a firm's output. It follows directly from Euler's theorem for homogeneous functions that within the interior of $S$,
\begin{equation}
    F(x)=x\cdot\frac{\partial F}{\partial x}.
\end{equation}
Since the output of the production function is a measure of the firm's capacity to pay for the needed resources, we see that if the price of resource $i$, denoted $p_i$, is set according to 
\begin{equation} \label{eq:price}
    p_{i}=\frac{\partial F}{\partial x_{i}},
\end{equation}
for all $i \in [M]$, then the output of the firm suffices exactly to purchase all the resources required, and there is no surplus profit. This is known as competitive equilibrium, which maximizes social welfare in the sense that the price of each good is commensurate to its usefulness in increasing the total output \cite{Arrow1951-ih, Debreu1959-di}. 

While it has long been understood at a qualitative level that data is an inherently different resource than the ones considered above due to the ability to copy it for free \cite{Schumpeter1942-jr, Arrow1962-hf}, the quantitative form of this statement was realized decades later by the seminal work of Romer \cite{Romer1990-sn}. If we include data $y$ as an input into the production function, we instead have 
\begin{equation} \label{eq:prod_function_w_data}
    F(\lambda x, y)=\lambda F(x, y)
\end{equation}
rather than the expected need to double each input proportionate to match production as in $F(\lambda x, \lambda y) = \lambda F(x, y)$. This is because the data used by one process can be copied and used by several with negligible additional cost. Euler's theorem once again gives 
\begin{equation}
    F(x,y)=x\cdot\frac{\partial F}{\partial x}.
\end{equation}
However, since increasing the amount of data will generally increase the output (say by improving the quality of inference), we have $\frac{\partial F}{\partial y}>0$. It follows that 
\begin{equation}
    F(x,y)<x\cdot\frac{\partial F}{\partial x}+y\frac{\partial F}{\partial y}.
\end{equation}
Due to this inequality, it is impossible to set prices according to \cref{eq:price}. If this were done for all inputs including data, the total output would be insufficient to pay for all the required resources. As a result, markets involving data must be inherently inefficient in the sense that one must underpay for some resource, or must include some external mechanism to enforce adequate compensation for resources that can be freely replicated. Mechanisms such as patent law or subsidies that incentivize innovation are all examples of this. Other examples are afforded by the trusted third parties that are introduced in proposals for data markets and handle the data in lieu of the data buyers themselves \cite{Agarwal2019-fv}. In the context of strategic games that model data selling, the ability to copy data is also manifest in the payoff for the data seller being independent of the number of buyers, unless a mechanism is put in place by which the data buyers all agree to pay in advance for their data \cite{Nageeb_Ali2020-hv}.

\subsubsection{Consumable data as a factor of production}

We can interpret the results of \Cref{sec:examples} within this framework (at the limit of large $m,N$ so that $m$ can be considered to be a continuous variable, and computing derivatives with respect to it becomes meaningful). Taking the linear regression sampling problem as an example, the solution of $\mathrm{MLRS}_{N,m}$ is analogous to the output of a production function, with the number of samples $m$ and Alice's message equivalent to $\lambda$ and $y$ respectively.
The result of \Cref{lem:MLRS_cb} is then analogous to \cref{eq:prod_function_w_data}. Up to constant factors, this is an example of the well-known nonrival nature of classical data. Alice must send a significant portion of her input to Bob for him to produce even a single sample, and once Alice sends her full input he can produce an unlimited number of samples in this way. If Alice were to sell Bob her data in the setting of a strategic game, her potential payoff will be essentially independent of the value that Bob can derive (since this is proportional to $m$). 

On the other hand, \Cref{lem:MLRS_qb} indicates that if Alice insists on using quantum communication, the data is analogous to a rival good as described by \cref{eq:prod_function}. Bob can still produce $m$ samples, but this requires that Alice sends at least a number of qubits proportional to $m$. If Alice were to charge Bob for each qubit sent for example, she would obtain a payoff proportional to the Bob's output $m$ (as long as $m < N$). The lower bound indicates that this scaling holds regardless of the strategy Alice uses to encode her input into the message, and of the strategy Bob uses to process this message. Using classical resources alone this would be impossible to achieve. We make these notions more precise in the context of a strategic game that models a data market in \Cref{sec:auction}. 

A similar analogy can be made with respect to the Multiple Hidden Matching problem and the multi-party observable estimation problem. 


\subsection{A posted price data auction with consumable data} \label{sec:auction}

We would like to identify more concretely the economic consequences of the consumable nature of quantum data. 
We consider a formulation naturally related to auction theory \cite{Krishna2009-de, Roughgarden2016-xp}. Alice's action space $A_A = \bR_+$ is the set of prices she charges for a single bit or qubit of her input. Once Alice fixes a price $p$, Bob is free to purchase as many bits/qubits as he wants. Bob's action space is thus $A_B = \bN$, and we denote the number he purchases by $b$. This is known as a posted price auction with only a single bidder and multiple items (or a particularly simple combinatorial auction). 
Assume the number of samples $m$ takes values in $[\overline{m}]$ and Alice has no knowledge of it (say she holds a uniform prior). We also assume the matrices $B_i$ are chosen in a worst-case fashion (in order for our communication lower bounds to be applicable). 

For any values of $m,p,b$, the payoffs of the two players are 
\begin{equation}
    v_{A}(m,p,b)=pb,\quad v_{B}(m,p,b)=\#\text{S}(m,b)-pb,
\end{equation}
where $\#\text{S}(m,b)$ represents the number of samples Bob can produce using a message of $b$ bits/qubits, given that he holds $m$ such $B_i$). 

Consider first the quantum communication case. We know from our lower bound \Cref{lem:MLRS_qb} that for sufficiently large $m$, there is an absolute constant $C$ such that, if Bob were to purchase $b$ qubits produced by Alice, then 
\begin{equation}
    \#\text{S}^{Q}(m,b)\leq\frac{Cb}{\log(N/\#\text{S}^{Q}(m,b))}\approx\frac{Cb}{\log(N)}
\end{equation}
for some absolute constant $C$. We also assume $N \gg m$ which allows us to use the approximation $\log N - \log \# \text{S}^{Q}(m,b) \approx \log N$ since this slightly simplifies the analysis.
Since additionally $\#\text{S}(m,b)\leq m$ by definition, we have the upper bound 
\begin{equation}
    v^Q_{B}(m,p,b)\leq\min\left\{ \frac{Cb}{\log N},m\right\} -pb.
\end{equation}
If we also assume that Bob's payoff is maximized at the point $b^\star$ that maximizes this upper bound, he is interestsed in solving
\begin{equation}
\underset{b}{\max}\min\left\{ \frac{Cb}{\log N},m\right\} -pb=\left\{ \begin{array}{ccc}
m(1-\frac{p\log N}{C})\ \  & 0 \leq p<\frac{C}{\log N} & (b^{\star}=\frac{m\log N}{C})\\
0\ \  & p\geq\frac{C}{\log N} & (b^{\star}=0)
\end{array}\right.
\end{equation}
with the corresponding value of $b^\star$ in the right column. 
Alice's payoff is maximized by thus choosing $p$ as close as possible to $C/\log N$ from below without exceeding it, and will be equal to $b^\star(m,p) p = mp\log (N)/C=\tilde{
\Omega}(m)$. This holds for any $m$ for which \Cref{lem:MLRS_qb} holds, even though Alice has no knowledge of $m$.


In the classical case, we know the problem is non-consumable from \Cref{lem:random}. This implies that for $m=1$, there is a message of length $\kappa$ independent of $m$ which Alice can send, which Bob can then re-use to produce say $\rho m$ samples with some constant probability, for some $\rho \leq 1$. 

This implies
\begin{equation}
    v_{B}^{C}(m,p,b)=\mathbf{1}\left[b\geq\kappa\right]\rho m-pb.
\end{equation}
Bob thus solves 
\begin{equation}
\underset{b}{\max}\mathbf{1}\left[b\geq\kappa\right]\rho m-pb=\left\{ \begin{array}{ccc}
\rho m-p\kappa & 0\leq p<\frac{\rho m}{\kappa} & (b^{\star}=\kappa)\\
0 & p\geq\frac{\rho m}{\kappa} & (b^{\star}=0)
\end{array}\right.
\end{equation}
Note that unlike the quantum case, Alice has no way of knowing how to choose $p$ appropriately ahead of time, since the critical value below which she receives no payoff depends on $m$. If she wants to guarantee a nonzero payoff she has to choose $p=\rho/\kappa$ (i.e. assume $m=1$) in which case her payoff is independent of $m$. 



\section{Open questions}
\label{sec:open}

Our work raises the following open questions in communication complexity:

\begin{enumerate}
    \item Can the lower bound on the one-way quantum communication complexity of $\mathrm{MHM}_{N,m}$ be improved to $\Omega(m)$ or even $\Omega(m\log N)$ for $m \ll \sqrt{N}$? 
    \item Can the class of problems with a quantum asymmetric direct sum property be characterized in some generality?
    \item Are there explicit problems with asymmetric direct sums for randomized communication? For decision problems, the scaling can be at most $O(\log m)$. Can we get a $\Omega(m)$ scaling for relations?
    \item We considered the scenario where the graph of required computations is a \emph{star graph}, where one input must be paired with every other input. You can more generally consider the asymmetric direct sum for any \emph{bipartite graph}. It is easy to see that the proof of \cref{lem:MHM_QLB} generalizes to say that the quantum communication complexity scales with the sum of the square-roots of the degrees of the vertices on the left. Can this be improved?
\end{enumerate}

A remark regarding question 3 is that for a random relation, it is actually impossible to get a constant probability of success for \emph{every} Bob for superconstant number of Bobs because we do not have amplification of success probability. This is why we ask for explicit examples. The lower bound technique would also need to avoid the trivial upper bound if we only require a constant fraction of Bobs to succeed (\cref{lem:s-random}).



Our results are unconditional but restricted to specific problems and one-way communication. By making additional assumptions or utilizing the strategic nature of the problems we consider, it may be possible to extend the class of problems that enable consumable data. One possible direction is outlined below. Moreover, in \cref{apx:discussion} we explore the broader implications of our results beyond the scope of complexity theory.

\paragraph{Computational assumptions.} A setting we have not yet considered, that is touched upon by the task of shadow tomography, is one where the computational power of Bob is restricted. It has been noted that general shadow tomography procedures are expected to scale polynomially with the dimension of the Hilbert space of $\rho$ or the trial state $\rho_T$.  If Bob is restricted to polylog computational time, then the creation of the clonable hypothesis state may become impossible.  This is analogous to the effect in cryptographic no-cloning theorems on pseudorandom quantum states~\cite{Ji2018-cc}, where even when sample efficient cloning is possible, no computationally efficient scheme can be used to clone the states of interest. In this context, it is worth noting that learning of certain states that have efficient descriptions, such as pseudorandom states~\cite{Zhao2023-fs}, is known to be computationally hard. The addition of computational restrictions on Bob hence potentially widens the class of consumable data tasks, but requires moving beyond a communication complexity model that permits unbounded computation.

\subsection*{Acknowledgements}
The authors thank Scott Aaronson, Dmytro Gavinsky, and Or Sattath for insightful discussions and comments on the manuscript. S.J. was supported by Scott Aaronson's CIQC grant.

\bibliographystyle{alphaurl}
\bibliography{paperpile.bib, refs.bib}

\newcommand{\etalchar}[1]{$^{#1}$}
\begin{thebibliography}{CvDNT99}

\bibitem[Aar09]{Aaronson2009-uz}
Scott Aaronson.
\newblock Quantum copy-protection and quantum money.
\newblock In {\em 24th Annual IEEE Conference on Computational Complexity},
  pages 229--242. IEEE Computer Soc., Los Alamitos, CA, 2009.

\bibitem[Aar18]{Aaronson2018-in}
Scott Aaronson.
\newblock Shadow tomography of quantum states.
\newblock In {\em Proceedings of the 50th Annual ACM SIGACT Symposium on Theory
  of Computing}, pages 325--338, New York, NY, USA, 2018. ACM.

\bibitem[ADS19]{Agarwal2019-fv}
Anish Agarwal, Munther Dahleh, and Tuhin Sarkar.
\newblock A marketplace for data: An algorithmic solution.
\newblock In {\em Proceedings of the 2019 ACM Conference on Economics and
  Computation}, New York, NY, USA, 2019. ACM.

\bibitem[AJLA{\etalchar{+}}12]{Asharov2012-dq}
Gilad Asharov, Abhishek Jain, Adriana L\'{o}pez-Alt, Eran Tromer, Vinod
  Vaikuntanathan, and Daniel Wichs.
\newblock Multiparty computation with low communication, computation and
  interaction via threshold {FHE}.
\newblock In {\em Advances in Cryptology -- EUROCRYPT 2012}, Lecture notes in
  computer science, pages 483--501. Springer Berlin Heidelberg, Berlin,
  Heidelberg, 2012.

\bibitem[ALL{\etalchar{+}}21]{Aaronson2021-gq}
Scott Aaronson, Jiahui Liu, Qipeng Liu, Mark Zhandry, and Ruizhe Zhang.
\newblock New approaches for quantum copy-protection.
\newblock In {\em Advances in cryptology---CRYPTO 2021. Part I}, volume 12825
  of {\em Lecture Notes in Comput. Sci.}, pages 526--555. Springer, Cham, 2021.

\bibitem[AR19]{Aaronson2019-on}
Scott Aaronson and Guy~N Rothblum.
\newblock Gentle measurement of quantum states and differential privacy.
\newblock In {\em STOC'19---Proceedings of the 51st Annual ACM SIGACT Symposium
  on Theory of Computing}, pages 322--333. ACM, New York, 2019.

\bibitem[Arr51]{Arrow1951-ih}
Kenneth~J Arrow.
\newblock An extension of the basic theorems of classical welfare economics.
\newblock In {\em Proceedings of the second Berkeley symposium on mathematical
  statistics and probability}, volume~2, pages 507--533, 1951.

\bibitem[Arr62]{Arrow1962-hf}
Kenneth~J Arrow.
\newblock Economic welfare and the allocation of resources for invention.
\newblock In {\em The Rate and Direction of Inventive Activity}, pages
  609--626. Princeton University Press, Princeton, 1962.

\bibitem[ATS03]{Aharonov2003-sw}
Dorit Aharonov and Amnon Ta-Shma.
\newblock Adiabatic quantum state generation and statistical zero knowledge.
\newblock In {\em Proceedings of the Thirty-Fifth Annual ACM Symposium on
  Theory of Computing}, pages 20--29. ACM, New York, 2003.

\bibitem[BDS23]{BenDavid2023quantumtokens}
Shalev Ben-David and Or~Sattath.
\newblock Quantum {T}okens for {D}igital {S}ignatures.
\newblock {\em {Quantum}}, 7:901, January 2023.
\newblock \href {https://doi.org/10.22331/q-2023-01-19-901}
  {\path{doi:10.22331/q-2023-01-19-901}}.

\bibitem[BGS13]{BGS13}
Anne Broadbent, Gus Gutoski, and Douglas Stebila.
\newblock Quantum one-time programs (extended abstract).
\newblock In {\em Advances in cryptology---{CRYPTO} 2013. {P}art {II}}, volume
  8043 of {\em Lecture Notes in Comput. Sci.}, pages 344--360. Springer,
  Heidelberg, 2013.
\newblock URL: \url{https://doi.org/10.1007/978-3-642-40084-1_20}, \href
  {https://doi.org/10.1007/978-3-642-40084-1\_20}
  {\path{doi:10.1007/978-3-642-40084-1\_20}}.

\bibitem[BI20]{Broadbent2020-xw}
Anne Broadbent and Rabib Islam.
\newblock Quantum encryption with certified deletion.
\newblock In {\em Theory of cryptography. Part III}, volume 12552 of {\em
  Lecture Notes in Comput. Sci.}, pages 92--122. Springer, Cham, 2020.

\bibitem[BMR{\etalchar{+}}20]{Brown2020-uq}
Tom Brown, Benjamin Mann, Nick Ryder, Melanie Subbiah, Jared~D Kaplan, Prafulla
  Dhariwal, Arvind Neelakantan, Pranav Shyam, Girish Sastry, Amanda Askell,
  Sandhini Agarwal, Ariel Herbert-Voss, Gretchen Krueger, Tom Henighan, Rewon
  Child, Aditya Ramesh, Daniel Ziegler, Jeffrey Wu, Clemens Winter, Chris
  Hesse, Mark Chen, Eric Sigler, Mateusz Litwin, Scott Gray, Benjamin Chess,
  Jack Clark, Christopher Berner, Sam McCandlish, Alec Radford, Ilya Sutskever,
  and Dario Amodei.
\newblock Language models are few-shot learners.
\newblock {\em Advances in Neural Information Processing Systems},
  33:1877--1901, 2020.

\bibitem[BO21]{Badescu2021-kb}
Costin B\u{a}descu and Ryan O'Donnell.
\newblock Improved quantum data analysis.
\newblock In {\em STOC '21---Proceedings of the 53rd Annual ACM SIGACT
  Symposium on Theory of Computing}, pages 1398--1411. ACM, New York, 2021.

\bibitem[BRWY13]{Braverman2013-to}
Mark Braverman, Anup Rao, Omri Weinstein, and Amir Yehudayoff.
\newblock Direct products in communication complexity.
\newblock In {\em 2013 IEEE 54th Annual Symposium on Foundations of Computer
  Science}, pages 746--755. IEEE, 2013.

\bibitem[BYJK08]{Bar-Yossef2008-uf}
Ziv Bar-Yossef, T~S Jayram, and Iordanis Kerenidis.
\newblock Exponential separation of quantum and classical one-way communication
  complexity.
\newblock {\em SIAM J. Comput.}, 38(1):366--384, 2008.

\bibitem[Can20]{canonne2020shortnotelearningdiscrete}
Clément~L. Canonne.
\newblock A short note on learning discrete distributions, 2020.
\newblock URL: \url{https://arxiv.org/abs/2002.11457}, \href
  {https://arxiv.org/abs/2002.11457} {\path{arXiv:2002.11457}}.

\bibitem[CCHL22]{Chen2022-kx}
Sitan Chen, Jordan Cotler, Hsin-Yuan Huang, and Jerry Li.
\newblock Exponential separations between learning with and without quantum
  memory.
\newblock In {\em 2021 IEEE 62nd Annual Symposium on Foundations of Computer
  Science---FOCS 2021}, pages 574--585. IEEE Computer Soc., Los Alamitos, CA,
  2022.

\bibitem[CKP25]{CKP25}
Rohit Chatterjee, Srijita Kundu, and Supartha Podder.
\newblock Are uncloneable proof and advice states strictly necessary?
\newblock In {\em S{TOC}'25---{P}roceedings of the 57th {A}nnual {ACM}
  {S}ymposium on {T}heory of {C}omputing}. ACM, New York, [2025] \copyright
  2025.

\bibitem[CLLZ21]{CLLZ21}
Andrea Coladangelo, Jiahui Liu, Qipeng Liu, and Mark Zhandry.
\newblock Hidden cosets and applications to unclonable cryptography.
\newblock In {\em Advances in Cryptology – CRYPTO 2021: 41st Annual
  International Cryptology Conference, CRYPTO 2021, Virtual Event, August
  16–20, 2021, Proceedings, Part I}, page 556–584, Berlin, Heidelberg,
  2021. Springer-Verlag.
\newblock \href {https://doi.org/10.1007/978-3-030-84242-0_20}
  {\path{doi:10.1007/978-3-030-84242-0_20}}.

\bibitem[CTJ{\etalchar{+}}21]{Chen2021-hn}
Mark Chen, Jerry Tworek, Heewoo Jun, Qiming Yuan, Henrique Ponde de~Oliveira
  Pinto, Jared Kaplan, Harri Edwards, Yuri Burda, Nicholas Joseph, Greg
  Brockman, Alex Ray, Raul Puri, Gretchen Krueger, Michael Petrov, Heidy
  Khlaaf, Girish Sastry, Pamela Mishkin, Brooke Chan, Scott Gray, Nick Ryder,
  Mikhail Pavlov, Alethea Power, Lukasz Kaiser, Mohammad Bavarian, Clemens
  Winter, Philippe Tillet, Felipe~Petroski Such, Dave Cummings, Matthias
  Plappert, Fotios Chantzis, Elizabeth Barnes, Ariel Herbert-Voss,
  William~Hebgen Guss, Alex Nichol, Alex Paino, Nikolas Tezak, Jie Tang, Igor
  Babuschkin, Suchir Balaji, Shantanu Jain, William Saunders, Christopher
  Hesse, Andrew~N Carr, Jan Leike, Josh Achiam, Vedant Misra, Evan Morikawa,
  Alec Radford, Matthew Knight, Miles Brundage, Mira Murati, Katie Mayer, Peter
  Welinder, Bob McGrew, Dario Amodei, Sam McCandlish, Ilya Sutskever, and
  Wojciech Zaremba.
\newblock Evaluating large language models trained on code.
\newblock {\em arXiv [cs.LG]}, 2021.

\bibitem[CvDNT99]{Cleve1999-dy}
Richard Cleve, Wim van Dam, Michael Nielsen, and Alain Tapp.
\newblock Quantum entanglement and the communication complexity of the inner
  product function.
\newblock In {\em Quantum computing and quantum communications (Palm Springs,
  CA, 1998)}, volume 1509 of {\em Lecture Notes in Comput. Sci.}, pages 61--74.
  Springer, Berlin, 1999.

\bibitem[Deb59]{Debreu1959-di}
Gerard Debreu.
\newblock {\em Theory of value: An axiomatic analysis of economic equilibrium},
  volume~17.
\newblock Yale University Press, 1959.

\bibitem[DR14]{Dwork2014-gq}
Cynthia Dwork and Aaron Roth.
\newblock The algorithmic foundations of differential privacy.
\newblock {\em Found. Trends Theor. Comput. Sci.}, 9(3-4):211--407, 2014.

\bibitem[EKR18]{Evans2018-py}
David Evans, Vladimir Kolesnikov, and Mike Rosulek.
\newblock A pragmatic introduction to secure multi-party computation.
\newblock {\em Found. Trends\textregistered{} Priv. Secur.}, 2(2-3):70--246,
  2018.

\bibitem[GA23]{pmlr-gong23a}
Weiyuan Gong and Scott Aaronson.
\newblock Learning distributions over quantum measurement outcomes.
\newblock In Andreas Krause, Emma Brunskill, Kyunghyun Cho, Barbara Engelhardt,
  Sivan Sabato, and Jonathan Scarlett, editors, {\em Proceedings of the 40th
  International Conference on Machine Learning}, volume 202 of {\em Proceedings
  of Machine Learning Research}, pages 11598--11613. PMLR, 23--29 Jul 2023.
\newblock URL: \url{https://proceedings.mlr.press/v202/gong23a.html}.

\bibitem[GIKR01]{Gennaro2001-xw}
Rosario Gennaro, Yuval Ishai, Eyal Kushilevitz, and Tal Rabin.
\newblock The round complexity of verifiable secret sharing and secure
  multicast.
\newblock In {\em Proceedings of the thirty-third annual ACM symposium on
  Theory of computing}, New York, NY, USA, 2001. ACM.

\bibitem[GKR08]{Goldwasser2008-gy}
Shafi Goldwasser, Yael~Tauman Kalai, and Guy~N Rothblum.
\newblock {One-Time} programs.
\newblock In {\em Lecture Notes in Computer Science}, Lecture notes in computer
  science, pages 39--56. Springer Berlin Heidelberg, Berlin, Heidelberg, 2008.

\bibitem[GLR{\etalchar{+}}24]{cryptoeprint:2024/1934}
Aparna Gupte, Jiahui Liu, Justin Raizes, Bhaskar Roberts, and Vinod
  Vaikuntanathan.
\newblock Quantum one-time programs, revisited.
\newblock Cryptology {ePrint} Archive, Paper 2024/1934, 2024.
\newblock URL: \url{https://eprint.iacr.org/2024/1934}.

\bibitem[GM24]{cryptoeprint:2024/1798}
Sam Gunn and Ramis Movassagh.
\newblock Quantum one-time protection of any randomized algorithm.
\newblock Cryptology {ePrint} Archive, Paper 2024/1798, 2024.
\newblock URL: \url{https://eprint.iacr.org/2024/1798}.

\bibitem[GMW87]{Goldreich1987-na}
O~Goldreich, S~Micali, and A~Wigderson.
\newblock How to play {ANY} mental game.
\newblock In {\em Proceedings of the nineteenth annual ACM conference on Theory
  of computing - STOC '87}, New York, New York, USA, 1987. ACM Press.

\bibitem[GS19]{Gosset2019-di}
David Gosset and John Smolin.
\newblock A compressed classical description of quantum states.
\newblock In {\em 14th Conference on the Theory of Quantum Computation,
  Communication and Cryptography}, volume 135 of {\em LIPIcs. Leibniz Int.
  Proc. Inform.}, pages Art. No. 8, 9. Schloss Dagstuhl. Leibniz-Zent. Inform.,
  Wadern, 2019.

\bibitem[HBC{\etalchar{+}}22]{Huang2022-ot}
Hsin-Yuan Huang, Michael Broughton, Jordan Cotler, Sitan Chen, Jerry Li, Masoud
  Mohseni, Hartmut Neven, Ryan Babbush, Richard Kueng, John Preskill, and
  Jarrod~R McClean.
\newblock Quantum advantage in learning from experiments.
\newblock {\em Science}, 376(6598):1182--1186, 2022.

\bibitem[HK17]{Hazan2017}
Itay Hazan and Eyal Kushilevitz.
\newblock Two-party direct-sum questions through the lens of multiparty
  communication complexity.
\newblock In {\em 31 {I}nternational {S}ymposium on {D}istributed {C}omputing},
  volume~91 of {\em LIPIcs. Leibniz Int. Proc. Inform.}, pages Art. No. 26, 15.
  Schloss Dagstuhl. Leibniz-Zent. Inform., Wadern, 2017.

\bibitem[HKP20]{Huang2020-gf}
Hsin-Yuan Huang, Richard Kueng, and John Preskill.
\newblock Predicting many properties of a quantum system from very few
  measurements.
\newblock {\em Nature Physics}, 16(10):1050--1057, 2020.

\bibitem[JK21]{Jain2021-of}
Rahul Jain and Srijita Kundu.
\newblock A direct product theorem for one-way quantum communication.
\newblock In {\em 36th Computational Complexity Conference}, volume 200 of {\em
  LIPIcs. Leibniz Int. Proc. Inform.}, pages Art. No. 27, 28. Schloss Dagstuhl.
  Leibniz-Zent. Inform., Wadern, 2021.

\bibitem[JLS18]{Ji2018-cc}
Zhengfeng Ji, Yi-Kai Liu, and Fang Song.
\newblock Pseudorandom quantum states.
\newblock In {\em Lecture Notes in Computer Science}, Lecture notes in computer
  science, pages 126--152. Springer International Publishing, Cham, 2018.

\bibitem[JT20]{Jones2020-fv}
Charles~I Jones and Christopher Tonetti.
\newblock Nonrivalry and the economics of data.
\newblock {\em Am. Econ. Rev.}, 110(9):2819--2858, 2020.

\bibitem[KN96]{Kushilevitz_Nisan_1996}
Eyal Kushilevitz and Noam Nisan.
\newblock {\em Communication Complexity}.
\newblock Cambridge University Press, 1996.

\bibitem[Kre95]{Kremer1995-kp}
Ilan Kremer.
\newblock {\em Quantum communication}.
\newblock PhD thesis, Hebrew University of Jerusalem, 1995.

\bibitem[Kri09]{Krishna2009-de}
Vijay Krishna.
\newblock {\em Auction Theory}.
\newblock Academic Press, San Diego, CA, 2 edition, 2009.

\bibitem[KS95]{Kurz1995-yj}
Heinz~D Kurz and Neri Salvadori.
\newblock {\em Theory of Production: A Long-Period Analysis}.
\newblock Cambridge University Press, 1995.

\bibitem[LS09]{Lee09}
Troy Lee and Adi Shraibman.
\newblock Lower bounds in communication complexity.
\newblock {\em Foundations and Trends® in Theoretical Computer Science},
  3(4):263--399, 2009.
\newblock URL: \url{http://dx.doi.org/10.1561/0400000040}, \href
  {https://doi.org/10.1561/0400000040} {\path{doi:10.1561/0400000040}}.

\bibitem[LSS08]{Lee2008-ss}
Troy Lee, Adi Shraibman, and Robert Spalek.
\newblock A direct product theorem for discrepancy.
\newblock In {\em 2008 23rd Annual IEEE Conference on Computational
  Complexity}, pages 71--80. IEEE, 2008.

\bibitem[MBS{\etalchar{+}}23]{Merchant2023-pm}
Amil Merchant, Simon Batzner, Samuel~S Schoenholz, Muratahan Aykol, Gowoon
  Cheon, and Ekin~Dogus Cubuk.
\newblock Scaling deep learning for materials discovery.
\newblock {\em Nature}, 624(7990):80--85, 2023.

\bibitem[MS24]{Montanaro2024-qj}
Ashley Montanaro and Changpeng Shao.
\newblock Quantum communication complexity of linear regression.
\newblock {\em ACM Trans. Comput. Theory}, 16(1):Art. 1, 30, 2024.

\bibitem[NACZL20]{Nageeb_Ali2020-hv}
S~Nageeb~Ali, Ayal Chen-Zion, and Erik Lillethun.
\newblock Reselling information.
\newblock {\em arXiv [cs.GT]}, 2020.

\bibitem[Nay99]{Nayak1999-if}
Ashwin Nayak.
\newblock Optimal lower bounds for quantum automata and random access codes.
\newblock In {\em 40th Annual Symposium on Foundations of Computer Science (New
  York, 1999)}, pages 369--376. IEEE Computer Soc., Los Alamitos, CA, 1999.

\bibitem[NC10]{Nielsen2010-jj}
Michael~A Nielsen and Isaac~L Chuang.
\newblock {\em Quantum Computation and Quantum Information: {10th} Anniversary
  Edition}.
\newblock Cambridge University Press, 2010.

\bibitem[Raz99]{Raz1999-ic}
Ran Raz.
\newblock Exponential separation of quantum and classical communication
  complexity.
\newblock In {\em Proceedings of the thirty-first annual ACM symposium on
  Theory of Computing}, STOC '99, pages 358--367, New York, NY, USA, 1999.
  Association for Computing Machinery.

\bibitem[RK11]{Regev2011-wq}
Oded Regev and Bo'az Klartag.
\newblock Quantum one-way communication can be exponentially stronger than
  classical communication.
\newblock In {\em Proceedings of the forty-third annual ACM symposium on Theory
  of computing}, New York, NY, USA, 2011. ACM.

\bibitem[Rom90]{Romer1990-sn}
Paul~M Romer.
\newblock Endogenous technological change.
\newblock {\em J. Polit. Econ.}, 98(5):S71--S102, 1990.

\bibitem[Rou16]{Roughgarden2016-xp}
Tim Roughgarden.
\newblock {\em Twenty lectures on algorithmic game theory}.
\newblock Cambridge University Press, Cambridge, England, 2016.

\bibitem[Sch42]{Schumpeter1942-jr}
Joseph~A Schumpeter.
\newblock {\em Capitalism, Socialism, and Democracy}.
\newblock Harper \& Brothers, New York, 1942.

\bibitem[Sha03]{Shaltiel2003-mx}
Ronen Shaltiel.
\newblock Towards proving strong direct product theorems.
\newblock {\em Comput. Complex.}, 12(1-2):1--22, 2003.

\bibitem[She11]{Sherstov2011-yk}
Alexander~A Sherstov.
\newblock Strong direct product theorems for quantum communication and query
  complexity.
\newblock In {\em Proceedings of the forty-third annual ACM symposium on Theory
  of computing}, New York, NY, USA, 2011. ACM.

\bibitem[VSH{\etalchar{+}}22]{Villalobos2022-vw}
Pablo Villalobos, Jaime Sevilla, Lennart Heim, Tamay Besiroglu, Marius
  Hobbhahn, and Anson Ho.
\newblock Will we run out of data? an analysis of the limits of scaling
  datasets in machine learning.
\newblock {\em arXiv [cs.LG]}, 2022.

\bibitem[Wie83]{Wiesner1983-lt}
Stephen Wiesner.
\newblock Conjugate coding.
\newblock {\em SIGACT News}, 15(1):78--88, 1983.

\bibitem[Yao82]{Yao1982-yt}
Andrew~C Yao.
\newblock Protocols for secure computations.
\newblock In {\em 23rd Annual Symposium on Foundations of Computer Science
  (sfcs 1982)}, pages 160--164. IEEE, 1982.

\bibitem[Yao86]{Yao1986-rw}
Andrew Chi-Chih Yao.
\newblock How to generate and exchange secrets.
\newblock In {\em 27th Annual Symposium on Foundations of Computer Science
  (sfcs 1986)}, pages 162--167. IEEE, 1986.

\bibitem[ZLK{\etalchar{+}}23]{Zhao2023-fs}
Haimeng Zhao, Laura Lewis, Ishaan Kannan, Yihui Quek, Hsin-Yuan Huang, and
  Matthias~C Caro.
\newblock Learning quantum states and unitaries of bounded gate complexity.
\newblock {\em arXiv [quant-ph]}, 2023.

\end{thebibliography}

\appendix

\section{Omitted Proofs}

    \begin{proof}[Proof of \Cref{lem:MHM_Dlb}] \label{pf:lem:MHM_Dlb}
        We begin by proving $D^\rightarrow(\mathrm{MHM}_{N,1}) \geq N/2+1$.
    
        A deterministic protocol $\cP$ for $\mathrm{MHM}_{N,1}$ is defined by a matrix with $2^N$ rows denoting the inputs to Alice and $(N-1)!!$ columns denoting the inputs to Bob ($(N-1)!!$ is the number of perfect matchings over $[N]$). The entry in the matrix corresponding to inputs $(x,M)$ is a tuple $(i,j,b)$ such that $(i,j) \in M$ and $b = x_i \oplus x_j$. Define by $\tau$ a message sent by Alice, and by $S_\tau$ the subset of the rows for which Alice sends $\tau$ to Bob. The choice of $(i,j,b)$ depends on $x$ only through the message $\tau$. 
        Since the protocol is deterministic, for a given column, the entries in each column of $S_\tau$ must have the same value since they share the same $\tau,M$, so we may write (with slight abuse of notation)
        \begin{equation}
            \mathcal{P}(x,M)=\mathcal{P}(\tau,M)=(i,j,b), \quad (i,j) \in M.
        \end{equation}
        Thus the rows of $S_\tau$ are all identical, and we can view each entry as a constraint that each vector $x$ for which Alice sends the message $\tau$ must obey. We will bound the maximal possible size of $S_\tau$ by bounding the number of $x$s that can satisfy all these constraints. 

        The constraints on the bits can be thought of as edges on a graph $G = (V,E)$ with nodes $V$ indexed by $[N]$. We begin with $E=\emptyset$ and choose a sequence of matchings $\cM = \{M^{\ell}\}$. For every matching, $\cP$ must produce a valid output that selects an edge from the matching and constrains the corresponding entries of $x$. While we have no control over which edge is chosen, we will choose $\cM$ in such a way that at each step of the algorithm, the size of the connected components in $G$ increases for any edge output by $\cP$. 
        
        Denote by $\{C^\ell_i\}$ the connected components of $G$ at step $\ell$, and $C^\ell = \underset{i}{\cup}C^\ell_{i}$. Initially we thus have $\left|C^0\right|=0$. 
        \begin{enumerate}[label=\roman*), align=left]
            \item $|C^\ell| \leq N/2$
            
                    We start with an arbitrary matching $M^{1}$. For any $x$ for which Alice communicates $\tau$, the entries in $S_\tau$ in the column corresponding to $M^{1}$ is  
                    $S_\tau$ must contain an edge $(i,j) \in M^1$, hence after adding $(i,j)$ to $E$ and $M^1$ to $\cM$ we have $|C^1|=2$. Denoting by $D^\ell$ the disconnected nodes, we next define a matching $M^2$ that pairs each node in $C^1$ with some node in $D^1$. The remaining nodes of $D^1$ are paired among themselves. Note that $M^2$ cannot be equal to $M^1$, since $M^1$ contained an edge between two nodes that are both in $C^1$ while $M^2$ does not. We add $(i,j)$ to $E$ where $\cP(\tau, M^{2}) = (i,j,b)$. If the edge connects $C^1$ and $D^1$, then $|C^2|=3$. Otherwise, $|C^2|=4$. 
                    
                    We pick $M^3,\dots$ in the same fashion, defining $M^{\ell+1}$ by pairing each node in $C^\ell$ with a node in $D^\ell$ (and pairing the remaining nodes arbitrarily). This can be done as long as $|C^\ell| \leq N/2$. At every stage, we are guaranteed that $M^{\ell+1} \notin \cM$ by the same argument used for $M^2$, hence we are assured that it is a valid choice. 
                    
                    After at most $N/2-1$ such steps, we have either $|C^\ell|=N/2+1$ or $|C^\ell|=N/2+2$. From this point a different strategy is required, since there are not enough disconnected nodes in $D^\ell$ to pair with all the nodes in $C^\ell$. Subsequently, we order the nodes in $C^\ell$ by first ordering the connected components $\{C^\ell_i\}$ by size, with $C_0^\ell$ being the largest (or tied for the largest, breaking ties arbitrarily), and then arbitrarily ordering the nodes within each $C^\ell_i$. 
                    
            \item $|C^\ell| > N/2$ and $|C^\ell_0| \leq N/2$

                    Order the nodes in $C^\ell$ in the manner specified above. Denote by $R^\ell_-$ the first $N/2$ nodes in this ordering, and by $R^\ell_+$ the remaining $|C^\ell|-N/2$ nodes. Define the matching $M^{\ell+1}$ by first pairing each node in $R^\ell_+$ with a node in $R^\ell_-$ in descending order (i.e. starting with the nodes in $C_0^\ell$). Note that two nodes in the same connected component cannot be paired in this way. This is because, if this occurred for some connected component $C_i$, this would imply that either $|C_i| > |C_0^\ell|$ (since $C_i$ must have a node in $R^\ell_-$, the boundary between $R^\ell_-$ and $R^\ell_+$ divides $C_i$, so every node in the matching so far is in $C_i$, and we started the pairing in $R^\ell_-$ with the nodes in $C_0^\ell$ and went through all of them and reached $C_i$) contradicting the imposed ordering, or else $C_i=C_0^\ell$, in which case since some nodes in $C^\ell_0$ are also in $R^\ell_+$, we have $|C_0^\ell|>N/2$ and we terminate the algorithm. 
                    Having thus paired all the nodes in $R^\ell_+$ (we can always do this since $|R^\ell_+|\leq N/2$), we complete $M^{\ell+1}$ by pairing the remaining nodes in $R^\ell_-$ with the unconnected nodes $D^\ell$ in an arbitrary way. Note that $M^{\ell+1}$ does not contain any edge between two nodes that are in the same connectivity component. Thus it is distinct from all of the matchings already in $\cM$ (since by construction each one contained such an edge) and we can add it to $\cM$. We add $(j,k)$ to $E$ where $\cP(\tau, M^{\ell+1}) = (j,k,b)$.  
                    
                    For the same reason specified above, the edge from $M^{\ell+1}$ that is selected by $\cP$ will either connect two previously unconnected components in $C^\ell$ hence $C^{\ell+1}_k=C^{\ell}_i\cup C^{\ell}_j$ for some $i,j,k$, or else connect some $C^\ell_i$ with a previously unconnected edge (meaning $|C^{\ell+1}|=|C^{\ell}|+1$).
                    
                    We run the above algorithm until some step $\tilde{\ell}$ when either (a) $|C^{\tilde{\ell}}|=N$ or (b) $C_0^{\tilde{\ell}}>N/2$. 
        \end{enumerate}
        The algorithm is guaranteed to terminate in $O(N)$ steps. If (a) occurs, then either (a1) there are strictly less than $N/2$ connectivity components or (a2) there are exactly $N/2$ connectivity components, since each one contains at least two nodes. In case (a1), there are strictly less than $N/2$ independent degrees of freedom in the choice of the bits of any $x$ for which Alice sends the message $\tau$, since each connectivity component $C^{\tilde{\ell}}_i$ implies $|C^{\tilde{\ell}}_i|$ constraints of the form $x_j \oplus x_k = b$ where $\cP(\tau, M) = (j,k,b), (j,k) \in M$ connects two nodes in $C^{\tilde{\ell}}_i$. In case (a2), there are $N/2$ connectivity components of size $2$. We then consider a final matching $M^{\tilde{\ell}+1}$ that first divides $\{C^{\tilde{\ell}}_i\}$ into groups of two $\{K_i\}$ and then pairs each node to a node in a different connectivity component within the same $K_i$. As before, this matching is valid since $M^{\tilde{\ell}+1} \notin \cM$. After including the edge in $\cP(\tau,M^{\tilde{\ell}+1})$ into $E$, $G$ will contain $N/2-1$ connected components. As before, there are strictly less than $N/2$ degrees of freedom in choosing $x$. In case (b), there is a single component of size strictly larger than $N/2$. Thus even if all the remaining nodes are disconnected, there are strictly less than $N/2$ degrees of freedom once again. 

        In conclusion, in all cases we obtain that the number of rows of $S_\tau$ is at most $2^{N/2-1}$. The number of possible messages Alice must send is therefore at least $2^N/2^{N/2-1}=2^{N/2+1}$ and thus the number of bits Alice must send in order to solve $\mathrm{MHM}_{N,1}$ is at least $N/2+1$. Since this bound is valid for the multi-Bob version of the problem as well, we have $D^\rightarrow(\mathrm{MHM}_{N,m}) \geq N/2+1$. 

        The upper bound is trivial: Alice sends the Bobs the first $N/2+1$ bits of her input. These are sufficient for the Bobs to compute the output for all $m$ matchings simultaneously. The result follows. 
    \end{proof}

\begin{proof}[Proof of \Cref{lem:MLRS_cb}] \label{pf:MLRS_cb}
    \begin{enumerate}[label=\roman*), align=left]
        \item Theorem 9 of \cite{Montanaro2024-qj}, applied to square matrices. The proof is based on lower bounds for distributed Fourier sampling.
        \item It follows from the ability of Alice to send her whole input to Bob to complete the task.
    \end{enumerate}
\end{proof}

\begin{proof}[Proof of \Cref{lem:MLRS_qb}] \label{pf:MLRS_qb}
\begin{enumerate}[label=\roman*), align=left]
    \item 
    Say Alice is given a binary vector $y$ of length $m\log(N/m)$ and there are $m$ Bobs. Each Bob uses the matrix   
    \begin{equation}
        B_j = \underset{i=(N/m)j}{\overset{(N/m)(j+1)}{\sum}}\left|i\right\rangle \left\langle i\right|.
    \end{equation}
    Alice then divides her bits into $m$ sets of size $\log(N/m)$ and treats the bits in each set as an integer $r_j \in [N/m]$. She creates a vector $x$ of length $N$ by concatenating a unary encoding of these numbers, meaning
        \begin{equation}
        [x_{[(N/m)j:(N/m)(j+1)]}]_i=\sqrt{\frac{1}{m}}\delta_{ir_{j}},
    \end{equation}
    where we used $x_{[l:m ]}$ denotes the subset of the entries of a vector ranging from $[l, m)$.  

    Suppose Alice and the Bobs manage to solve $\mathrm{MLRS}_{N,m}$ with inaccuracy $\eta$. This means that Bob produces a sample from a distribution that is at most $\eta$ in TV from each of his target distributions $\cP_j$. From the definition of $x$ and the $B^{(j)}$, $\cP_j$ is be a delta function at $r_j$. This means that with probability at least $1-2\eta$, Bob recovers the $\log(N/m)$ bits of $r_j$ by performing a computational basis measurement. 
    It follows that Alice's message to Bob is a random-access encoding of $m\log(N/m)$ bits. 
    From known lower bounds on the number of qubits needed for random access coding \cite{Nayak1999-if}, if $2 \eta < 1/2$, Alice must send at least $\Omega(m\log(N/m))$ qubits to the Bobs. 
 
    \item This follows immediately from the bound of Theorem 4 of \cite{Montanaro2024-qj} with an additional factor of $m$ due to the number of samples, and using $||x||_2=1$. The bound uses an amplitude-encoding of $x$, followed by the application of $B^+_k$ using block-encoding. If two-way communication is allowed, the complexity can be improved to $O(m\log(N)\underset{k}{\max}\left\Vert B_{k}^{+}\right\Vert /\left\Vert B_{k}^{+}x\right\Vert _{2})$ since Alice and Bob can run amplitude amplification.
\end{enumerate}

\end{proof}

\section{Discussion}
\label{apx:discussion}
We demonstrated that there exist problems for which encoding classical data into quantum states leads to behavior that is akin to that of rival, or consumable, goods, which is generally not possible using classical data alone. The inherent privacy benefits of amplitude-encoded data might also facilitate computation with proprietary data, giving users fine-grained control over the dissemination of their private data without the need for additional encryption. The setup we consider also does not require end-users to possess a quantum computer in order to be valuable. Instead, the user must simply trust an entity possessing a networked quantum computer to distribute data states on their behalf. This is similar to entrusting a bank to distribute funds on the behalf of an account holder. While our results are based on communication complexity, they rely on the properties of the data encoding itself, and thus are also relevant in a scenario where different parties are provided access to the same quantum memory at different times, without requiring networked quantum computers.

Being a preliminary investigation into the possibility of using quantum networks in this manner, our results do not immediately apply to problems with clear economic value. If this were the case, it could enable novel types of data markets and incentive structures for the production of data. It is worth noting however that our results for the linear regression sampling problem apply also to a related problem in which Bob obtains a state that encodes the solution to a linear system rather than a classical sample. Such states are known to be strictly more powerful resources than classical samples \cite{Aharonov2003-sw}, and could potentially be useful in learning tasks such as updating the value of a linear estimator with new data (which is typically achieved with the recursive least squares algorithm).


The form of the quantum communication lower bound that indicates the rival behavior of quantum data is reminiscent of a direct sum theorem. Direct sum theorems demonstrate that the complexity of solving $m$ independent instances of certain problems scales linearly with $m$. They have been studied extensively in both the classical \cite{Shaltiel2003-mx, Braverman2013-to, Lee2008-ss} and quantum \cite{Sherstov2011-yk, Jain2021-of} setting. These results are not directly applicable since in our setting the inputs to Alice are not independent. Thus, this work motivates an \emph{asymmetric} direct sum result for classes of communication relations. 

In analogy to the potential clonability of quantum states with structure, there is a sense in which any non-consumable data may be cloned with respect to a particular task sample efficiently, even when cloning the overall state containing the information remains sample inefficient.  This is exemplified by the shadow tomography task above in which the task is solved via the creation of a classical representation of a hypothesis $\rho_T$, such that $\tr(E_i \rho_T) \approx \tr(E_i \rho)$ for all $i$ for the ground truth state $\rho$.  This classical representation $\rho_T$ need not be close in trace distance such that $||\rho - \rho_T||_{\text{tr}}$ is small, as would be required for a high fidelity cloning of the true state.  However it suffices for the task of shadow tomography, and admits an entirely classical representation that may be cloned through classical communication at will, making the data non-consumable, hence this task is clonable even when the underlying states might not be.

In restricting access to data that is used for computation, the setting we consider bears some resemblance to that of differential privacy~\cite{Dwork2014-gq, Aaronson2019-on}. In differential privacy, a query is promised not to reveal too much about individual datapoints. This is typically achieved classically by adding noise to data, while we achieve a similar capability in spirit by using a noiseless encoding into quantum states. 

\end{document}